\documentclass{article}

\title{A note on Harris' ergodic theorem, controllability and perturbations of harmonic networks}
\author{Renaud Raqu{\'e}pas}
\date{}

\usepackage[T1]{fontenc}
\usepackage{amsmath,amsfonts,amssymb}
\usepackage[margin=1.25in]{geometry}
\usepackage{pstool}
\usepackage{braket}
\usepackage{amsthm}
\usepackage{tikz}
\usepackage{xcolor}
\usepackage[colorlinks = true, linkcolor = blue, urlcolor  = teal, citecolor = violet]{hyperref}
\newtheorem{lemma}{Lemma}[section]
\newtheorem{theorem}[lemma]{Theorem}
\newtheorem{proposition}[lemma]{Proposition}

\theoremstyle{definition}

\newtheorem{remark}[lemma]{Remark}

\newtheorem{example}[lemma]{Example}

\newcommand{\numset}[1]{\mathbf{#1}}
	\newcommand{\cc}{\numset{C}}
	\newcommand{\rr}{\numset{R}}
	
	\newcommand{\zz}{\numset{Z}}
	\newcommand{\nn}{\numset{N}}

	\newcommand{\one}{\mathbf{1}}
	\newcommand{\e}{\mathrm{e}}
		\newcommand{\Exp}[1]{\mathrm{e}^{#1}}
	
	\makeatletter
	\providecommand*{\diff}%
		{\@ifnextchar^{\DIfF}{\DIfF^{}}}
	\def\DIfF^#1{%
		\mathop{\mathrm{\mathstrut d}}%
			\nolimits^{#1}\gobblespace}
	\def\gobblespace{%
		\futurelet\diffarg\opspace}
	\def\opspace{%
		\let\DiffSpace\!%
		\ifx\diffarg(%
			\let\DiffSpace\relax
		\else
			\ifx\diffarg[%
				\let\DiffSpace\relax
			\else
				\ifx\diffarg\{%
					\let\DiffSpace\relax
				\fi\fi\fi\DiffSpace}

	\providecommand*{\od}[3][]{%
		\frac{\diff^{#1}#2}{\diff #3^{#1}}}

	\renewcommand{\d}{\diff}

	\renewcommand{\L}{\mathcal{L}}

	\newcommand{\pp}{\mathbf{P}}
	\newcommand{\ee}{\mathbf{E}}

	\DeclareMathOperator{\tr}{tr}

\newcommand{\invar}{^\textnormal{inv}}

\newcommand{\power}{a}
\newcommand{\xinit}{x^\textnormal{in}}
\newcommand{\scball}{G}
\newcommand{\scballti}{D}

\begin{document}

\maketitle

\begin{center}
	\begin{tabular}{c c c c}
		Univ. Grenoble Alpes & &&  McGill University\\
		CNRS, Institut Fourier & && Dept. of Mathematics and Statistics\\
		F-38\,000 Grenoble  & && 1005--805 rue Sherbrooke O. \\
		France & &&  Montr\'eal (Qu\'ebec){\,\ }H3A 0B9 \ \ Canada \\
	\end{tabular}
\end{center}

\medskip

\begin{abstract}
	We show that elements of control theory, together with an application of Harris' ergodic theorem, provide an alternate method for showing exponential convergence to a unique stationary measure for certain classes of networks of quasi-harmonic classical oscillators coupled to heat baths. With the system of oscillators expressed in the form $$\d X_t = A X_t \d t + F(X_t) \d t + B \d W_t $$ in~$\rr^d$, where~$A$ encodes the harmonic part of the force and~$-F$ corresponds to the gradient of the anharmonic part of the potential, the hypotheses under which we obtain exponential mixing are the following: $A$ is dissipative, the pair $(A,B)$ satisfies the Kalman condition, $F$ grows sufficiently slowly at infinity (depending on the dimension~$d$), and the vector fields in the equation of motion satisfy the weak H\"ormander condition in at least one point of the phase space.
\end{abstract}

\section{Introduction}

Thermally driven networks of oscillators play an important role in the investigation of various aspects of nonequilibrium statistical mechanics. On a mathematical level, a driven network of classical harmonic oscillators can be modeled as a $d$-dimensional process~$(X_t)_{t \geq 0}$ described by a linear stochastic differential equation (\textsc{sde}) of the form
\begin{align*}
	\d X_t &= A X_t \d t + B \d Z_t,
\end{align*}
where the linear operators~$A$ and~$B$ satisfy certain structural conditions and where $(Z_t)_{t \geq 0}$ is a given $n$-dimensional stochastic process describing the noise due to thermal fluctuations. The integer~$n\leq d$ is the number of degrees of freedom of the network that are coupled to heat baths. The noise is often taken to be a Wiener process, but other types of noise are physically interesting. A particularly important question regarding such systems and perturbations thereof is that of invariant measures.

In this work, we consider~$A$ and~$B$ satisfying the \emph{Kalman condition}, a smooth globally Lipschitz perturbing vector field~$x \mapsto F(x)$ that grows slower than $|x|^{1/2d}$ at infinity\footnote{The power $\tfrac{1}{2d}$ is generically not optimal. As we will see, $d$ can be replaced by an integer $d_*$ appearing in the formulation of the Kalman rank condition. In all cases $d_* \leq d$.} and~$(W_t)_{t \geq 0}$ a Wiener process, and show with arguments from control theory and an application of Hairer and Mattingly's version of Harris' ergodic theorem that the process described by the \textsc{sde}
$$
	\d X_t = A X_t \d t + F(X_t) \d t + B \d W_t
$$
admits a unique stationary measure
when~$A$ is \emph{dissipitaive}
and a \emph{weak H\"ormander condition} on the vector fields in the \textsc{sde} holds in at least one point~$x_0$ of the phase space. Moreover, the convergence to this stationary measure then happens exponentially fast.
The abstract mathematical setup and the result are made more precise in Section~\ref{sec:setup}. The proof is provided in Section~\ref{sec:ab-res}.

\medskip

In Section~\ref{sec:net}, we introduce the mathematical description of  perturbed  networks of harmonic oscillators in this framework, both in the Langevin regime and in the so-called semi-Markovian regime, and for geometries that go beyond the 1-dimensional chain. In this context, the matrix~$A$ encodes the friction, kinetic and harmonic terms (both the pinning and the interaction) while the perturbation~$F$ corresponds to minus the gradient of the anharmonic part of the potential.

\medskip

In the case of a 1-dimensional chain of oscillators connected to heat baths at both ends, results of this type have been established for a very general class of quasi-homogeneous potentials~\cite{EPR99a,EPR99b,EH00,RBT02,Ca07}. The recent paper~\cite{CEHRB} extends these results to more complicated networks. Roughly speaking, these results require that the pinning potential grows as $|q|^{k_1}$ at infinity, that the interaction potential grows as $|q|^{k_2}$ with $k_2 \geq k_1 \geq 2$, and that the interaction part of the potential has no flat piece or infinitely degenerate points. While our growth condition is considerably more restrictive than the ones found in these works, the form of local nondegeneracy that we require is weaker: we only need a weak H\"ormander condition to hold at a single point. Moreover, our setup accommodates a wide variety of geometries and bounded many-body interaction terms (beyond pinning and two-body interactions).

Such results typically involve carefully studying smoothing properties of the associated Markov semigroup. The strategy here is different and instead relies on recent developments on the use of solid controllability in the study of mixing properties of random dynamical systems~\cite{AS05,AKSS07,Sh07,Sh17}. The simplicity of the argument can in itself justify the presentation of such an application.

Another advantage is that our general strategy is not based on the Gaussian structure of Brownian motion and can thus be more easily adapted to different types of noise that are physically relevant. Similar arguments can be used to discuss the analogous problem with compound Poisson processes; this type of problem will be analyzed in a subsequent work.

The proof can be summarized as follows. For a discrete-time Markov process, Harris' theorem states that the existence and uniqueness of an invariant measure, with exponentially fast convergence in the total variation metric, can be obtained from the existence of a suitable Lyapunov function and a minorization for the transition probabilities starting from any point in the interior of a suitable level set of that Lyapunov function. The precise statement we use is the one formulated in~\cite{HM11}; also see~\cite{Ha56} and~\cite{MeTw}. We then pass from discrete to continuous time.

The function $V(x) := \int_0^\infty |\Exp{s A} x|^2 \d s$ is shown to be a suitable Lyapunov function using dissipativity of~$A$, the behaviour of~$F$ at infinity, and basic It\^{o} calculus. The details are given in Section~\ref{sec:Lyap}.

In order to prove the lower bound on transitions, we use the Kalman condition on the pair~$(A,B)$ and again the estimate on the behaviour of~$F$ at infinity. These hypotheses yield that the point~$x_0$ in which the weak H\"ormander condition holds can be approached from $\{ V \leq R\}$ with a uniform lower bound on the probability. On the other hand, the weak H\"ormander condition in~$x_0$ implies solid controllability from~$x_0$ and we can combine solid controllability and approachability to obtain the desired lower bound. The details are given in Section~\ref{sec:contr}.

\medskip

Different sufficient conditions for the hypotheses of the main theorem to hold are given in more concrete terms throughout Sections~\ref{sec:net} and~\ref{sec:Hor}. In the former, we give criteria for the dissipativity, Kalman and growth conditions in terms of more physical quantities for networks of oscillators based on~\cite{JPS17}. In the latter, we give a perturbative condition for the weak H\"ormander condition to hold.

\paragraph*{Acknowledgements} The author would like to thank Armen Shirikyan for introduction to these questions and crucial suggestions for this particular application, No\'e Cuneo and Vojkan Jak\v{s}i\'c for informative discussions, as well as the D\'epartement de math\'ematiques at Universit\'e Cergy--Pontoise, where part of this research was conducted, for its hospitality. The research of the author was funded by the Natural Sciences and Engineering Research Council of Canada (NSERC), the Fonds de recherche du Qu\'ebec\,--\,Nature et technologies (FRQNT) and the NonStops project of the Agence nationale de la recherche  (ANR-17-CE40-0006-02).

\section{Setup, assumptions and main result}\label{sec:setup}

\paragraph{Notation}
Throughout the paper, we use:
	$\|{\,\cdot\,}\|$ to denote the operator norm of linear maps;
~$\{e_i\}_{i=1}^n$ for the standard orthonormal basis of~$\rr^n$;
    $|{\,\cdot\,}|$ to denote the euclidean norm on $\rr^d$ (arising from the standard inner product~$\braket{{\,\cdot\,},{\cdot\,}}$);
  	 $B(x,r)$ for the open ball of radius $r > 0$ centered at the point~$x$ in~$\rr^d$;
	 $C^k_0([0,T];\rr^n)$ to denote the space of~$k$~times continuously differentiable functions~$\eta : [0,T] \to \rr^n$ with $\eta(0) = 0$;
	 $\operatorname{Prob}(\rr^d)$ for the space of Borel probability measures on~$\rr^d$;
	 $\mathcal{L}_G$ for the Lie derivative with respect to the vector field~$G$;
	 $\one_S$ to denote the indicator function of the set~$S$.
	 The natural numbers~$\nn$ start at~1. The underlying probability space is $(\Omega,\mathcal{F},\pp)$ and we use the letter~$\omega$ to denote elementary events there.

\medskip

Let $d$ and $n$ be natural numbers with $n \leq d$ and let $\omega \mapsto (W_t(\omega))_{t \geq 0}$ be a Wiener process in~$\rr^n$. We are interested in the $d$-dimensional diffusion process $\omega \mapsto (X_t(\xinit,\omega))_{t \geq 0 }$ governed by the equation
\begin{align}
	X_t(\xinit,\omega) = \xinit + \int_0^t A X_s(\xinit,\omega) + F(X_s(\xinit,\omega))\d s + B W_t(\omega) \label{eq:SDE}
\end{align}
where~$B: \rr^n \to\rr^d$ is a linear map, $A: \rr^d \to \rr^d$ is a linear map,~$F$ is a smooth globally Lipschitz vector field on~$\rr^d$, and~$\xinit\in\rr^d$ is an initial condition. We often omit writing explicitly the dependence on~$\xinit$ or~$\omega$ and write the equation in differential notation.
We assume the following dissipativity and controllability conditions on the linear maps~$A$ and~$B$.
\begin{enumerate}
	\item[\textbf{(D)}] the  eigenvalues of the linear map~$A$ (considered over~$\cc^d$) each have strictly negative real part.
	\item[\textbf{(K)}] the pair $(A,B)$ satisfies the \emph{Kalman condition}, meaning that the columns of $B$, $AB$, $A^2 B$, $A^3B$ and so forth span~$\rr^d$.

	Then, by the Cayley--Hamilton theorem, there exists $d_* \leq d$ such that $$\operatorname{span}\{Be_i, ABe_i, A^2 Be_i, \dotsc, A^{d_*-1}Be_i : i=1,\dotsc,n\} = \rr^d.$$
\end{enumerate}
The Kalman condition is commonly used in the basic theory of controllability for linear systems ({i.e.} when $F\equiv 0$); it is then equivalent to several notions of controllability~\cite[\S\S{1.2--1.3}]{Cor}.

We further assume that the perturbing vector field~$F$ satisfies the following growth condition.
\begin{enumerate}
	\item[\textbf{(G)}] there exists a constant~$\power \in [0,\tfrac{1}{2d_*})$ such that
	\begin{equation}
      \sup_{x \in \rr^d} \frac{|F(x)|}{(1+|x|)^\power} < \infty. \label{eq:F-growth}
	\end{equation}
\end{enumerate}
Finally, we suppose the existence of a point~$x_0$ where the \emph{weak H\"ormander condition} on the vector fields appearing in the stochastic equation~\eqref{eq:SDE} is satisfied.
\begin{enumerate}
  \item[\textbf{(H)}] there exists a point $x_0 \in \rr^d$ in which the family
	$$
		\{
			V_0, \mathcal{L}_{V_2} V_1, \mathcal{L}_{V_3} \mathcal{L}_{V_2} V_1, \dotsc : V_0 \in \mathcal{B} \text{ and } V_1, V_2, V_3, \dotsc \in \mathcal{B} \cup \{A + F\}
		\}
	$$
  of vector fields spans~$T_{x_0} \rr^d \cong \rr^d$, where $\mathcal{B} = \{B e_1, \dotsb, B e_n\}$.
\end{enumerate}

\begin{remark}
	In the linear case ({i.e.} when $F \equiv 0$), a straightforward computation shows that the Kalman condition~\textnormal{(K)} implies the weak H\"ormander condition~\textnormal{(H)}. This suggests that the latter can be obtained from a perturbative argument in a point~$x_0$ far from the origin if~$F$ can be neglected at infinity in a suitable sense; see Section~\ref{sec:Hor}.
\end{remark}

It is convenient to study the properties of such a diffusion process through the corresponding controlled equation
\begin{align}
   \begin{cases}\label{eq:pert-sys}
	\quad \dot x(t)  = A x(t) + F(x(t)) + B \dot\eta (t), \\
	\quad x(0) = \xinit,
   \end{cases}
\end{align}
understood as
\begin{align*}
   x(t) = \xinit + \int_0^t A x(s) + F(x(s)) \d s + B (\eta(t)-\eta(0))
\end{align*}
when $\eta$ is a merely continuous function. We define, for $0 \leq t \leq T$,
\begin{align*}
	S^F_t : \rr^d \times C_0([0,T]; \rr^n) &\to \rr^d \\
		(\xinit, \eta) &\mapsto x(t)
\end{align*}
giving the solution at time~$t$ of this problem. We refer to the second argument as the \emph{control}. The function $S_t^F$ is uniformly continuous in each argument. It is also Fr\'echet differentiable. We will make use of these regularity properties in Section~\ref{sec:contr}.

\begin{remark}\label{rem:dec}
	The law for~$\eta \in C_0([0,T];\rr^n)$ corresponding to the Wiener process~$W_t(\omega)$ restricted to the interval~$[0,T]$ in~\eqref{eq:SDE}, which we denote by~$\ell$, is \emph{decomposable} in the following sense.

	There exist a sequence~$(F_N)_{N \in \nn}$ of nested finite-dimensional subspaces and a sequence $(F'_N)_{N \in \nn}$ of closed subspaces of the Banach space~$C_0([0,T];\rr^n)$ such that
	\begin{itemize}
		\item[(i)] the union~$\bigcup_{N\in\nn} F_N$ is dense in~$C_0([0,T];\rr^n)$;
		\item[(ii)] the space $C_0([0,T];\rr^n)$ decomposes as the direct sum $F_N \oplus F'_N$ for each $N \in \nn$, with corresponding (bounded) projections~$\Pi_N$ and~$\Pi'_N$, and the measure~$\ell$ decomposes as the product $\ell_N \otimes \ell'_N$ of its projected measures;
		\item[(iii)] the projected measure $\ell_N$ possesses a smooth positive density~$\rho_N$ with respect to the Lebesgue measure on the finite-dimensional space~$F_N$.
	\end{itemize}
	The requirement of~\cite{Sh17} that $\Pi_N \zeta \to \zeta$ in norm does not hold for all controls~$\zeta \in C_0([0,T];\rr^n)$. However, the convergence will hold true on nice enough subsets\,---\,which suffices for our endeavour. These decomposability properties play a central role in the arguments of~\cite{Sh07,Sh17} and are discussed here in Appendix~\ref{sec:app}.
\end{remark}

We use $P^F_t(\xinit,\cdot\,)$ to denote the distribution of the random variable~$\omega \mapsto X_t(\xinit,\omega)$ defined by~\eqref{eq:SDE}.
Then, $P^F_t$ satisfies the Chapman--Kolmogorov equation:
$$
   P^F_T(x,\Gamma) =\int_{\rr^d} P^F_{T-t}(y,\Gamma) P^F_t(x,\d y)
$$
for all times $0 \leq t \leq T$, all $x\in\rr^d$ and all Borel sets~$\Gamma \subseteq \rr^d$. We are interested in the large-time behaviour of~$P^F_t$. Our main result is the following.

\begin{theorem}\label{thm:main-new}
	Suppose that the \textsc{sde}
	\begin{align*}
		\d X_t  &= A X_t \d t + F(X_t)\d t + B \d W_t \label{eq:SDE}
	\end{align*}
	satisfies the conditions~\textnormal{(D)},~\textnormal{(K)}, \textnormal{(G)} and~\textnormal{(H)}.
	Then,
	it admits a unique invariant measure~$\mu\invar \in \operatorname{Prob}(\rr^d)$.
	Moreover, the function $V : \rr^d \to [0,\infty)$ defined by
	$$
		x \mapsto \int_{0}^\infty \braket{\Exp{s A} x, \Exp{s A} x} \d s
	$$
	is integrable with respect to~$\mu\invar$ and there exist constants $c,C > 0$ such that
	\begin{equation}\label{eq:exp-conv}
		\Big| \int_{\rr^d} f(y) P^F_t(\xinit,\d y) - \int_{\rr^d} f(y) \mu\invar(\d y) \Big|
			\leq C(1 + V(\xinit)) \Exp{-ct}
	\end{equation}
	for all $\xinit \in \rr^d$, all $t \geq 0$ and all measurable functions~$f$ with $|f| \leq 1 + V$.

\end{theorem}

The proof of this theorem is developed throughout Section~\ref{sec:ab-res}. The last key step there is an application of Hairer and Mattingly's version of Harris' ergodic theorem~\cite{HM11}. It requires two hypotheses: the existence of constants $\gamma \in (0,1)$ and $K > 0$ such that
\begin{equation}\label{eq:fsh-Lyap}
	\Big| \int_{\rr^d} V(y) P^F_t(x,\d y) \Big| \leq \gamma^t V(x) + K
\end{equation}
for all~$x \in\rr^d$ and all $t\geq 0$, and the existence of a positive measure bounding from below the probability of reaching a set when starting from the interior of a suitable level set of~$V$:
\begin{equation}\label{eq:fsh-nu-below}
	P^F_T(x,{\cdot\,}) \geq \nu_T
\end{equation}
for all $x \in \rr^d$ such that $V(x) \leq 1 + {2K}({1 - \gamma})^{-1}$.
The first one is dealt with in Section~\ref{sec:Lyap}; the second one, in Section~\ref{sec:contr}.

\section{Proof of Theorem~\ref{thm:main-new}}\label{sec:ab-res}

\subsection{Dissipativity and Lyapunov stability}\label{sec:Lyap}
	Condition~\textnormal{(D)} ensures that the integral defining $V :
			x \mapsto \int_{0}^\infty |\Exp{s A} x|^2 \d s
	$
	converges. To this function $V$ is naturally associated a positive definite matrix~$M$ such that $V(x) = \braket{x, M x}$.  We wish to show that, under the conditions~\textnormal{(D)} and~\textnormal{(G)}, this function satisfies the inequality~\eqref{eq:fsh-Lyap}
	for some constants $\gamma \in (0,1)$ and $K > 0$ that do not depend on~$x$.

	\begin{lemma}\label{cor:lyap}
			Under the conditions~\textnormal{(D)} and~\textnormal{(G)}, there exist constants~$K > 0$ and~$\gamma \in (0,1)$ such that the function~$V$ satisfies
			\begin{equation*}
				\Big|\int_{\rr^d} V(y) P^F_t(x,\d y)\Big| \leq \gamma^t V(x) + K
			\end{equation*}
			for all~$x \in \rr^d$ and all $t \geq 0$.
	\end{lemma}

	\begin{proof}
		Fix an initial condition $X_0 \in \rr^d$. First note that we have
		$$
		\braket{D_x V(x), A x} = 2\braket{x, MA x} = \int_{0}^\infty \od{}{s} \braket{\Exp{s A} x, \Exp{s A} x} \d s
		= -|x|^2.
		$$
		On the other hand, by assumption~\textnormal{(G)}, there exists $c_1 > 0$ such that $|F(x)| \leq \tfrac{1}{8\|M\|}|x| + c_1 $ and thus there exists a constant~$c_2>0$ depending on~$c_1$ and~$\|M\|$ such that
		$$
			\braket{D_x V(x), A x + F(x)} \leq -\tfrac 12 |x^2| + c_2.
		$$
		for all $x \in \rr^d$.

		By It\^o's lemma applied to the smooth function~$V$ (with no explicit $t$-dependence),
		\begin{align*}
			\d V(X_t)
			&= \braket{D V(X_t), A X_t + F(X_t)} \d t + 2 \braket{M X_t, B \, \d W_t} + \tr(M B B^*) \d t
		\end{align*}
		and thus
		\[
			\ee V(X_t) \leq V(X_0) + \int_0^t (-\tfrac{1}{2}\ee |X_s|^2 + c_2) \d s + \tr(M B B^*)t.
		\]
		Since $\Exp{s A}$ is nonsingular for any~$s \in [0,1]$ by assumption~\textnormal{(D)}, there exists $c_3 > 0$ depending on the eigenvalues of~$A$ such that
		$$
		 V(x) \geq \int_0^1 |\Exp{s A} x|^2 \geq c_3 |x|^2
		$$
		for all $x \in \rr^d$. Hence,
		\[
			\ee V(X_t) \leq V(X_0) - \int_0^t \tfrac{1}{2 c_3}\ee V(X_s) \d s + (c_2 + \tr(M B B^*))t
		\]
		By Gr\"onwall's inequality, we conclude that there exists a constant $K > 0$ (independent of~$X_0$) such that
		\[
			\ee V(X_t) \leq \Exp{-\frac{t}{2c_3}} V(X_0) + K. \qedhere
		\]
	\end{proof}

\subsection{Approachability and solid controllability} \label{sec:contr}
	The goal of this section is to show the existence of a time~$T>0$ and a nontrivial measure~$\nu_T$ on~$\rr^d$ such that the bound
	\[
		P^F_{T}(x, {\cdot\,}) \geq \nu_T
	\]
	holds for all $x\in\rr^d$ such that $V(x) \leq 1 + 2K(1-\gamma)^{-1}$, where~$\gamma$ and~$K$ are as in Lemma~\ref{cor:lyap}. This is done in two steps: we first control the probability of reaching neighbourhoods of~$x_0$ where~\textnormal{(H)} holds, and then the probability of reaching an arbitrary set when starting from~$x'$ close enough to~$x_0$.

	Throughout this section, the controlled nonlinear system~$\eqref{eq:pert-sys}$ is to be thought of as a perturbation of the controlled linear system
	\begin{align}
      \begin{cases}\label{eq:non-pert-sys}
		\quad \dot z(t)  = A z(t) + B \dot\eta (t ),  \\
		\quad z(0) = \xinit.
   \end{cases}
	\end{align}
	For $\eta \in C_0([0,T];\rr^n)$ and $0 \leq t \leq T$, $S_T(\xinit,\eta)$ is defined as the solution at time~$t$ of the problem~\eqref{eq:non-pert-sys}.

	We set $R := 1 + {2K}({1 - \gamma})^{-1}$.
	We make extensive use of the compact set~$\{x\in \rr^d : V(x) \leq R\}$, which we often write as $\{V \leq R\}$ for short.

	We start by showing that the point~$x_0$ in which the weak H\"ormander condition \textnormal{(H)} holds can be approximately reached with suitable control when starting from $\{V \leq R\}$.\footnote{This part of the argument actually holds for any~$x_0 \in \rr^d$, regardless of the H\"ormander condition.} To do this, we need a technical lemma on a matrix often referred to as the \emph{controllability Gramian}, which is used to construct relevant controls; see {e.g.}~\cite[\S\S{1.2--1.3}]{Cor}.

	\begin{lemma}\label{lem:norm-Qinv}
		If $A$ and $B$ are such that the Kalman condition~\textnormal{(K)} is satisfied with $d_*$, then the symmetric positive definite matrix
		\[
			Q_T = \int_0^T \Exp{t A} B B^* \Exp{t A^*} \d t
		\]
		has full rank and
		$
			\|Q_T^{-1}\| = O(T^{1-2d_*})
		$
		as $T \to 0$.
	\end{lemma}

	\begin{proof}
		Because $Q_T$ is symmetric and by real-analyticity of the maps $(0,1) \ni T \mapsto \braket{x, Q_T x} \in \rr_+$, it suffices to show that for each $x \in \rr^d$ with~$|x| = 1$, there exists $k \leq 2 d_* - 1$ such that
		\[
		 \partial_T^k \braket{x, Q_T x} |_{T=0}\neq 0.
		\]
		Suppose for contradiction that there exists such~$x$ with $|x|=1$ and
		$0 = \partial_T^k \braket{x, Q_T x}$ for each $k \leq 2 d_* -1$. From the first derivative, we have
		\[
			B^* x = 0.
		\]
		From the third derivative, we have
		\[
			\braket{x, B B^* (A^*)^2 x} +2\braket{x,  A B B^* A^*x} +\braket{x, A^2 B B^* x} = 0,
		\]
		but then, using again the consequence of the vanishing first derivative, we have
		\[
			B^* A^*x = 0.
		\]
		Inductively, from the $(2j+1)$th derivative, we have
		\[
			B^* (A^*)^j x = 0,
		\]
		for $j = 0, 1, \dotsc d_* - 1$. We conclude that $$x \in \bigcap_{j=0}^{d_*-1} \ker (B^* (A^*)^{j}) = \bigcap_{j=0}^{d_*-1} (\operatorname{ran} (A^j B))^{\perp},$$ contradicting the Kalman condition.
	\end{proof}

	\begin{proposition}
		Fix $x_0 \in \rr^d$. If the growth condition~\textnormal{(G)} and the Kalman condition~\textnormal{(K)} hold, then for any $x \in \rr^d$, $\delta > 0$ and $T > 0$ there exists a control~$\eta_{x,\delta,T} \in C^1_0([0,T];\rr^n)$ such that $S^F_T(x, \eta_{x, \delta, T}) \in B(x_0,\tfrac 12 \delta)$.
	\end{proposition}

	\begin{proof}
		Let $x \in \rr^d$ and $\delta > 0$ be arbitrary. Because the Kalman condition~\textnormal{(K)} holds, for any~$T \in (0,1]$, the control
		\[
			\zeta_{x,T}(t) := \int_0^t B^* \Exp{(T-s)A^*} Q_T^{-1} (x_0 - \Exp{-T A}x) \d s
		\]
		is such that $S_T (x, \zeta_{x,T}) = x_0$; see {e.g.}~\cite[\S{1.2}]{Cor}. We immediately have the bound
		\[
		|\dot \zeta_{x,T}(t)| \leq \|B\|\Exp{T\|A\|}\|Q_T^{-1}\|(|x_0| + \Exp{T\|A\|}|x|)
		\]
		and the hypotheses yield through Lemma~\ref{lem:norm-Qinv} the existence of a constant $C>0$ depending on~$A$ and~$B$ such that
		\[
		|\dot \zeta_{x,T}(t)| \leq C (|x| + |x_0|)T^{-m}
		\]
		for all $T \in (0,1]$, where $m := 2d_* - 1$.

		With $z_T(t) := S_t (x, \zeta_{x,T})$, $x_T(t) := S^F_t (x, \zeta_{x,T})$ and $y_T(t) := x_T(t) - z_T(t)$, we have
		\begin{align*}
			\dot y_T (t) &= A y_T(t) + F(x_T(t)), \\
			y_T(0) &= 0.
		\end{align*}
		Then, for $t \in [0,T]$,
		\begin{align*}
			y_T(t) &= \int_0^t \Exp{(t-s)A} F(x_T(s)) \d s
				= \int_0^t \Exp{(t-s)A} F(y_T(s) + z_T(s)) \d s.
		\end{align*}
		By~\textnormal{(G)}, there exists $C'>0$ depending on~$A$ and~$F$ only such that
		\begin{align*}
			|y_T(t)| \leq C' \int_0^t 1 + |y_T(s)|^\power + |z_T(s)|^\power \d s.
		\end{align*}
		On the other hand,
		\begin{align*}
			|z_T(t)| &\leq |\Exp{tA} x| + \int_0^t |\Exp{(t-s)A} B \dot\zeta_{x,T}(s)| \d s \\
				&\leq C'|x| + t C\Exp{T\|A\|}\|B\|(|x|+|x_0|)T^{-m}.
		\end{align*}
		Combining these two inequalities, there exists a constant~$C'' > 0$ such that
		\begin{align*}
			|y_T(t)|
				&\leq C'' \int_0^t |y_T(s)|\d s + t C'' (1 + |x| + |x_0|) (1 + T^{a(1-m)})
		\end{align*}
		Recall that $0 \leq \power < \frac{1}{2d_*}$ and $m = 2d_* + 1$. Hence,
		\[
			\power(1-m) + 1 > 0
		\]
		and, by Gr\"onwall's inequality, there exists $T_{x,\delta} \in (0,1]$ small enough, depending continuously on~$x$ and $\delta$, such that $|S^F_T(x,\zeta_{x,T}) - x_0| = |y_T(T)| < \tfrac 14 \delta$ for all $0 < T \leq T_{x,\delta}$.

		If $T \leq T_{x,\delta}$, pick $\eta_{x,\delta,T} = \zeta_{x,T}$. If $T > T_{x,\delta}$, let
		\[
			r_{T} := \sup_{0 \leq t \leq T} |S_t(x,0)| \quad \text{ and } \quad s_{T} = \min \{\tfrac 12 T, \inf_{|y| \leq r_{x,T}} T_{y,\delta}\}.
		\]
		Then, $|S_{T - s_T}(x,0)| < r_T$ and by the above $\zeta_{S_{T - s_T}(x,0), s_T}$ is such that
		\[
			S_{s_{T}}(S_{T - s_T}(x,0),\zeta_{S_{T - s_T}(x,0), s_T}) \in B(x_0,\tfrac 14 \delta).
		\]
		This corresponds to the control \[\tilde\eta_{x,\delta,T}(t) := \one_{[T-s_{T},T]}(t) \zeta_{S_{T - s_T}(x,0), s_T}(t-(T-s_{T}))\] defined on~$[0,T]$. A $C^1_0([0,T];\rr^n)$ regularisation~$\eta_{x,\delta,T}$ of~$\tilde\eta_{x,\delta,T}$ will then satisfy
		$
			S_T(x,\eta_{x,\delta,T}) \in B(x_0,\tfrac 12 \delta).
		$
	\end{proof}

	\begin{proposition}\label{cor:x-to-near-x0}\label{prop:appr}
		Fix $x_0 \in \rr^d$ and $\delta>0$ and suppose that the conditions \textnormal{(G)} and \textnormal{(K)} hold. Then, the function
		\[
			(x,T) \mapsto P^F_T(x,B(x_0,\delta))
		\]
		is positive and jointly lower semicontinuous.
	\end{proposition}

	\begin{proof}
		For any $x \in\rr^d$ and $T > 0$, there exists $\eta_{x,\delta,T} \in C^1_0([0,T];\rr^n)$ such that $S^F_T(x,\eta_{x,\delta,T}) \in B(x_0,\tfrac 12 \delta)$.
		By the Stroock--Varadhan support theorem,%
		\footnote{In the case of an additive noise, the Stroock--Varadhan support theorem can be given a direct proof by continuity arguments even if the vector field is unbounded, as long as the solutions are defined globally in time.}
		the support of the distribution of paths $[0,T] \ni t \mapsto X_t(x,\omega)$ contains the closure of
		$
	 		\{[0,T] \ni t \mapsto S_{t}^F(x,\eta) : \eta \in C^1_0([0,T];\rr^n)\}
		$
		with respect to the supremum norm on $C_0([0,T];\rr^d)$. In particular,
		$
			P_{T}^F(x, B(x_0,\delta)) > 0.
		$

		For $\pp$-almost every $\omega \in \Omega$, the path $t \mapsto W_t(\omega)$ is continuous. Since $X_t$ satisfies the integral equation
		\[
			\quad X_t(x,\omega) = x + \int_0^t AX_s(x,\omega) + F(X_s(x,\omega)) \d s + B W_t(\omega)
		\]
		with $y \mapsto Ay + F(y)$ globally Lipschitz and $t \mapsto B W_t(\omega)$ continuous, a standard argument shows that the map
		$
			(x,T) \mapsto X_T(x,\omega)
		$
		is jointly continuous. Therefore, the function
		\[
			(x,T) \mapsto \one_{\{\omega' \in \Omega \ : \ X_T(x,\omega') \in B(x_0,\delta)\}} (\omega)
		\]
		is jointly lower semicontinuous for $\pp$-almost all~$\omega \in \Omega$. Then, so is the map
		\[
			(x,T) \mapsto \int_\Omega \one_{\{\omega' \in \Omega \ : \ X_T(x,\omega') \in B(x_0,\delta)\}} (\omega) \d \pp(\omega)
		\]
		by Fatou's lemma.
	\end{proof}

	Now that we have established that, starting from~$\{V \leq R\}$, any neighbourhood of~$x_0$ can be suitably reached, we seek a minorization for transitions from points close to~$x_0$ to arbitrary points of the space. In~\cite{Sh17}'s study of \textsc{sde}s on compact manifolds, the notions of decomposability and \emph{solid controllability} are used to show that the {weak H\"ormander condition}~\textnormal{(H)} in~$x_0$ is sufficient to provide appropriate control of the transition probabilities from points~$x'$ close enough to~$x_0$.
  \begin{enumerate}
	  \item[\textbf{(sC)}] a system $S : \rr^d \times E \to \rr^d$, where~$E$ is a Banach space, is said to be \emph{solidly controllable} from~$x_0$, with compact $Q \Subset E$, if there is a ball~$\scball$ in~$\rr^d$ and a number~$\epsilon > 0$ such that if a continuous map $\Phi: Q \to \rr^d$ satisfies
	  \[
		  \sup_{\zeta \in Q} |\Phi(\zeta) - S(x_0, \zeta)| \leq \epsilon,
	  \]
	  then $\Phi(Q) \supseteq \scball$.
  \end{enumerate}
	Most of the ideas for the next three results are present in different parts of~\cite{Sh17}; also see~\cite{Sh07}. We retrieve the key steps and repiece them in a way that is suitable for our endeavour.

  \begin{lemma}\label{lem:red}
	  If there exists a closed ball $\scballti \Subset \rr^d$ and a continuous function $f : \scballti \to E$ such that $S(x_0,f(x)) = x$ for all $x \in \scballti$, then~$S$ satisfies the solid controllability condition~\textnormal{(sC)} from~$x_0$, with $Q = f(\scballti)$.
  \end{lemma}

  \begin{proof}
	  Take $\epsilon < \tfrac 14 \operatorname{diam}(\scballti)$ and set $\scball := \{x \in \scballti : d(x, \partial \scballti) \geq \epsilon \}$.
	  Let~$\Phi$ be a continuous map on~$f(\scballti)$ such that
	  \[
		  \sup_{\zeta \in f(\scballti)} |\Phi(\zeta) - S(x_0, \zeta)| \leq \epsilon.
	  \]
	  Then, for any $x' \in \scball$, the continuous function $\Psi_{x'}$ defined on~$\scballti$ by
	  \[
		  \Psi_{x'}(x) = x' - \Phi(f(x)) + x
	  \]
	  maps~$\scballti$ to itself. Indeed,
	  \begin{align*}
		  |x'-\Psi_{x'}(x)| &= |x'-( x' - \Phi(f(x)) + x)| \\
			  & = |\Phi(f(x))- S(x_0,f(x))|
			  \leq \sup_{\zeta \in f(\scballti)} |\Phi(\zeta) - S(x_0, \zeta)|
			  \leq \epsilon.
	  \end{align*}
	  Hence, by the Brouwer fixed point theorem, there exists~$x \in \scballti$ such that $x = \Psi_{x'} (x)$, {i.e.} such that~$x' = \Phi(f(x))$. We conclude $\scball \subseteq \Phi(f(\scballti))$. 
  \end{proof}

 	 We will use this for~$S_1^F$ defined in Section~\ref{sec:setup}. In this case, the Banach space~$E$ of controls is~$C_0([0,1];\rr^n)$ equipped with the supremum norm.

	\begin{proposition}
		If the weak H\"ormander condition~\textnormal{(H)} is satisfied in~$x_0$, then ~$S_1^F$ is solidly controllable from~$x_0$, with a set $Q$ consisting of functions that are all Lipschitz with a common Lipschitz constant~$\kappa$.
	\end{proposition}

	\begin{proof}
		By the previous lemma, to show solid controllability, it suffices to provide a ball $\scballti \Subset \rr^d$ and a continuous function $f : \scballti \to C_0([0,1];\rr^n)$ such that $S^F_1(x_0,f(x_*)) = x_*$ for all $x_* \in \scballti$.

		As part of Theorem~2.1 in~\cite[\S{2.2}]{Sh17}, it is shown in a similar setting that the H\"ormander condition implies the existence of a ball $\scballti \Subset \rr^d$ and a continuous function $\tilde f : \scballti \to L^2([0,1];\rr^n)$ such that the solution of
		\[
			\begin{cases}
				\dot x = A x + F(x) + B \, \tilde f(x_*) \\
				x(0) = x_0
			\end{cases}
		\]
		satisfies $x(1) = x_*$.
		Moreover,
		$
		 \kappa := \sup_{x_* \in \scballti} \|\tilde f(x_*)\|_{C_0} < \infty.
		$ The construction of~$\scballti$ and~$\tilde f$ uses local arguments and can be directly translated to our setup.

		\begin{itemize}
		\item[]
		{\small
		The idea behind the proof is the following. Consider the following extended problem for~$ y(t) = (x(t),s(t))$ in~$\rr^d \times \rr$:
		\begin{equation}\label{eq:ext-prob}
			\begin{cases}
					\dot y = (Ax + F(x), 1) + (B \xi , 0) \\
					y(0) = (x_0, 0)
			\end{cases},
		\end{equation}
		where the control~$\xi$ is taken in~$L^2([0,1];\rr^n)$. The H\"ormander condition implies that the Lie algebra generated by the family $\{ \tilde V_\eta(x,s) = (Ax + F(x), 1) + (B \eta , 0): \eta\in\rr^n\}$ of vector fields has full rank at the point $(x_0,0)$. Hence, one can show using ideas from the proof of Krener's theorem that there exists a choice of small intervals $(a_l, b_l) \subset [0,1]$ and vectors~$\eta_l \in \rr^n$ for $l = 0, 1, \dotsc, d$ such that the parallelepiped $$\tilde \Pi =\{\alpha = (\alpha_0, \alpha_1, \dotsc, \alpha_d) \in \rr^{d+1} : \alpha_l \in (a_l , b_l) \}$$ embeds into $\rr^d \times \rr$ via the map
		\begin{align*}
			\phi : \tilde\Pi &\to \rr^d\times \rr \\
				\alpha &\mapsto \big(\Exp{\alpha_d \tilde V_{\eta_d}} \circ \dotsb  \circ \Exp{\alpha_0 \tilde V_{\eta_0}}\big)(x_0,0).
		\end{align*}
		In other words, $\phi$ takes~$\alpha$ to the solution~$y$ at time $T_\alpha := \alpha_0 +\alpha_1+\dotsb + \alpha_d$ of the extended problem~\eqref{eq:ext-prob} with the control
		\begin{equation}\label{eq:alpha-eta}
			\xi_\alpha(t) = \one_{[0,\alpha_0)}(t) \eta_0 + \sum_{l=1}^d \one_{[\alpha_{0} + \dotsb + \alpha_{l-1}, \alpha_{0} + \dotsb + \alpha_{l-1} + \alpha_l)}(t) \eta_l.
		\end{equation}
		Fixing an $\hat\alpha \in \tilde\Pi$ with corresponding $T_{\hat\alpha} \in (0,1]$, one finds that the solutions at time~$T_{\hat\alpha}$ of the problem
		\[
			\begin{cases}
				\dot x = A x + F(x) + B \xi_\alpha \\
				x(0) = x_0
			\end{cases}
		\]
		provide a diffeomorphsim between a neighbourhood of $\hat\alpha$ in $\{\alpha \in \tilde\Pi : T_\alpha = T_{\hat\alpha}\}$ and an open set~$O \subset \rr^d$. Inverting this diffeomorphism, one finds a function that associates to each point~$x_* \in {O}$ a control~$\xi_{\alpha(x_*)} \in L^2([0,T_{\hat\alpha}];\rr^n)$ of the form~\eqref{eq:alpha-eta}. By construction,
		$$S_{T_{\hat\alpha}}^F\Big(x_0,\int_0^{\cdot} \xi_{\alpha(x_*)}(s)\d s\Big) = x_*$$
		for all~$x_*\in O$. A standard argument then allows to find a closed ball~$\scballti \Subset \rr^d$ and a continuous function~$\tilde f : \scballti \to L^2([0,1];\rr^n)$ such that
		$$S_{1}^F\Big(x_0,\int_0^{\cdot} ({\tilde f}(x_*))(s)\d s\Big) = x_*$$
		for all~$x_* \in \scballti$.
		The supremum~$\kappa$ is bounded by the sum of the~$|\eta_l|$ used in the construction of the embedding~$\phi$.
		}
		\end{itemize}

		Let $f : \scballti \to C_0([0,1];\rr^n)$ be defined by $f(x_*) := \int_0^{\cdot} (\tilde f(x_*))(s) \d s$. Then,
		\begin{align*}
			\|f(x_*)-f(x_{**})\|_{C_0} &= \sup_{t \in [0,1]} \Big| \int_0^{t} (\tilde f(x_*))(s) \d s - \int_0^{t} (\tilde f(x_{**}))(s) \d s \Big| \\
			&\leq \|\tilde f(x_{*})-\tilde f(x_{**})\|_{L^2}
		\end{align*}
		so that~$f$ is continuous.
		We conclude that~$S_1^F$ is solidly controllable from~$x_0$,
		with $Q = f(\scballti)$. The constant~$\kappa$ is a common Lipschitz constant for all functions in~$Q$.
	\end{proof}

	\begin{proposition}\label{lem:appr-solid-cont}
		If the weak H\"ormander condition~\textnormal{(H)} is satisfied in~$x_0$, then there exist $\delta_0 > 0$ and a nonzero Borel measure~$\tilde\nu$ on~$\rr^d$ such that
		\[
			P^F_1(x', {\cdot\,}) \geq \tilde\nu
		\]
		for all $x' \in B(x_0,\delta_0)$.
	\end{proposition}

	\begin{proof}
		By the previous proposition, we have solid controllability of the system~$S_1^F$ from the point~$x_0$, with a set~$Q$ consisting of Lipschitz functions.
		Then, the strategy of~\cite[\S{1.2}]{Sh17} (also see~~\cite[\S{2.1}]{Sh07}) yields the desired measure. We outline the argument for completeness and to emphasize that we do not need the full strength of the decomposability assumption made there.

		Let $\Pi_N$ be as in Remark~\ref{rem:dec} and Appendix~\ref{sec:app}. Because all controls in~$Q$ have a common Lipschitz constant~$\kappa$, we have
		\[
			\lim_{N\to\infty}\sup_{\zeta\in Q} \| \zeta - \Pi_{N} \zeta \|_{C_0} = 0
		\]
		by Lemma~\ref{lem:not-quite-bounded}.
		Then, because $S^F_1(x_0,\cdot\,) : C_0([0,1];\rr^n) \to \rr^d$ is uniformly continuous, there exists~$N \in \nn$ large enough that
		\begin{equation*}
			\sup_{\zeta\in Q} |S^F_1(x_0, \Pi_N \zeta) - S^F_1(x_0, \zeta)| < \epsilon,
		\end{equation*}
		for the~$\epsilon$ in~\textnormal{(sC)}.
		Taking $\Phi = S^F_1(x_0,\Pi_N\,\cdot\,)$ there, $\Phi(Q)$ contains a ball (which has positive measure).

		By Sard's theorem, there exists a point~$\zeta_0 \in Q$ in which $D\Phi$ has full rank. Because $\Phi \circ \Pi_N = \Phi$, this property still holds true if we restrict~$\Phi$ to~$F_N = \operatorname{ran} \Pi_N$. There then exists a $d$-dimensional subspace $F_N^1 \subseteq F_N$ such that~$D\Phi|_{\zeta_0}(F_N^1) = \rr^d$. Let $F_N^2$ be such that $F_N^1 \oplus F_N^2 = F_N$.
		We will write $\zeta \in F_N$ as $(\zeta^1,\zeta^2)$ according to this decomposition. More generally, we will write a generic element of~$C_0$ as $(\zeta^1,\zeta^2, \zeta')$ with $\zeta' \in F_N'$.
		The Jacobian of the map $S^F_1(x_0,(\,\cdot\,,\zeta_0^2,0)) : F_N^1 \to \rr^d$ at the point~$\zeta_0^1$ is a linear isomorphism between~$F_N^1$ and~$\rr^d$.

		By the implicit function theorem, there exist neighbourhoods~$V^1$ of~$\zeta_0^1$, $V^2$ of~$\zeta_0^2$, $V'$ of~$0$, $W$ of~$x_0$, $U$ of~$S_1^F(x_0,(\zeta_0^1,\zeta_0^2,0))$; and a continuously differentiable function
		$
			g : W \times U \times V^2 \times V' \to V^1
		$
		such that, for points in the appropriate open sets, $S_1^F(x',(\zeta^1,\zeta^2,\zeta')) = x_*$ is equivalent to $\zeta^1 = g(x',x_*,\zeta^2,\zeta')$.

		Recall that~$\ell$ equals the product measure~$\ell_N \times \ell_N'$ with $\ell_N$ possessing a continuous and positive density~$\rho_N$ on~$F_N$. Let $\chi : \rr^d \times C_0  \to [0,1]$ be continuous, supported in $W \times V^1 \times V^2 \times V'$, and equal to~1 at $(x_0,\zeta_0^1,\zeta_0^2,0)$. Then, for any Borel set $\Gamma \subseteq \rr^d$,
		\begin{align*}
				P_1^F(x', \Gamma) &\geq \iiint_{S_1^F(x',\cdot\,)^{-1}(\Gamma)}
				 \chi(x',\zeta^1,\zeta^2,\zeta')\rho_N(\zeta^1,\zeta^2)  \d\zeta^1 \d\zeta^2  \ell_N'(\d\zeta')  \\
				&=  \iint_{V^2 \times V'} \int_\Gamma  \frac{\chi(x',g(x',x_*,\zeta^2,\zeta'),\zeta^2,\zeta') \rho_N(g(x',x_*,\zeta^2,\zeta'),\zeta^2)}{\det[DS_1^F(x',(\,\cdot\,,\zeta^2,\zeta'))|_{g(x',x_*,\zeta^2,\zeta')}]} \d x_* \d\zeta^2 \ell_N'(\d\zeta')
		\end{align*}
		for all $x' \in W$.

		By continuity, there exist numbers~$\delta_0 > 0$ and~$\alpha > 0$ such that
		$$
			P_1^F(x', \Gamma) \geq \alpha \operatorname{vol}(\Gamma \cap B(S_1^F(x_0,\zeta_0),\delta_0))
		$$
		for all $x' \in B(x_0,\delta_0)$ and all Borel sets $\Gamma \subseteq \rr^d$. \qedhere
	\end{proof}

	Then, by the Chapman--Kolmogorov equation,
	\begin{align*}
		P^F_{T+1}(x, \Gamma)
			&\geq \int_{x' \in B(x_0,\delta_0)} P^F_{T}(x,\d x') P^F_1(x', \Gamma) \\
			&\geq \int_{x' \in B(x_0,\delta_0)} P^F_{T}(x,\d x') \tilde \nu(\Gamma)
			= P^F_{T}(x, B(x_0, \delta_0))\tilde\nu(\Gamma)
	\end{align*}
	for any Borel set $\Gamma \subseteq \rr^d$ and any $T > 0$. We conclude that for any $T > 1$ the nontrivial measure
	$$
		\nu_T := \Big(\inf_{x \in \{ V \leq R\}}P^F_{T-1}(x, B(x_0, \delta_0))\Big) \tilde\nu
	$$
	is such that
	\begin{equation*}
		P^F_T(x,{\cdot\,}) \geq \nu_T
	\end{equation*}
	for all $x \in \rr^d$ such that $V(x) \leq R$. The infimum in the definition of~$\nu_T$ is positive by Proposition~\ref{cor:x-to-near-x0}.

\subsection{Application of Harris' ergodic theorem}

Recall that, by Lemma~\ref{cor:lyap}, the conditions~\textnormal{(D)} and~\textnormal{(G)} ensure the existence of constants $K > 0$ and $\gamma \in (0,1)$ such that the function $V$ satisfies
\begin{equation}\label{eq:rcl-Lyap}
	\Big|\int_{\rr^d} V(y) P^F_t(x,\d y)\Big| \leq \gamma^t V(x) + K
\end{equation}
for all~$x \in \rr^d$ and all $t > 0$.
Using the conditions~\textnormal{(G)} and~\textnormal{(K)}, we also showed in Proposition~\ref{cor:x-to-near-x0} that, for any $\delta > 0$,
$
	(x,T) \mapsto P^F_{T}(x, B(x_0, \delta))
$
is positive and jointly lower semicontinuous. Then, we concluded from this, hypothesis~\textnormal{(H)} and the arguments of~\cite{Sh17} that, for any $T >1$, there is a nontrivial measure~$\nu_T$ such that
\begin{equation}\label{eq:nu-below}
	P^F_T(x,{\cdot\,}) \geq \nu_T
\end{equation}
for all $x \in \rr^d$ such that $V(x) \leq R$.

The existence of a function~$V$ satisfying the condition~\eqref{eq:rcl-Lyap} and a nontrivial measure~$\nu_T$ satisfying~\eqref{eq:nu-below} are precisely the hypotheses we need to apply Harris' theorem.

Indeed, considering the $T$-skeleton of our diffusion process\footnote{By $T$-skeleton of a (continuous time) stochastic process, we mean the restriction to times in the countable set~$T \nn$.} for $T = 2$, Theorem~1.2 in~\cite{HM11} yields constants $c,C >0$ and a stationary measure~$\mu\invar \in \operatorname{Prob}(\rr^d)$ against which $V$ is integrable and such that
\begin{align}\label{eq:m-bound-disc}
	\sup_{|f| \leq 1 + V}\Big| \int_{\rr^d} f(y) [P^F_{2m}(x,\d y) - \mu\invar(\d y)]\Big|
		&\leq
			C \Exp{-c(2m+2)}\big( 1 + V(x) \big)
\end{align}
for all $x \in \rr^d$ and all $m \in \nn \cup \{0\}$.

The measure $\mu\invar$ is the unique stationary probability measure for the $2$-skeleton, but it could {a priori} depend on our choice of $T$-skeleton. However, we can show that this measure is actually stationary, not only for the $2$-skeleton, but also for the continuous-time process.

Note that with $f = \one_\Gamma$ the indicator function of any Borel set $\Gamma \subseteq \rr^d$, integrating~\eqref{eq:m-bound-disc} in the variable~$x$ yields that
\begin{align}
	\Big|\int_{\rr^d} P^F_{2m}(x,\Gamma) \lambda(\d x) - \mu\invar(\Gamma) \Big| \leq C \Exp{-c(2m+2)} \Big( 1 + \int_{\rr^d} V(x) \lambda(\d x) \Big) \label{eq:tot-var-bound}
\end{align}
for any measure $\lambda \in \operatorname{Prob}(\rr^d)$.

Putting $\lambda$ defined by $\lambda(\Gamma) = \int P^F_s(x,\Gamma) \mu\invar(\d x)$ in~\eqref{eq:tot-var-bound} for some~$s\geq 0$, we have by the Chapman--Kolmogrov equation that
\begin{align*}
		&\Big|\int_{\rr^d} P^F_{2m+s}(x,\Gamma) \mu\invar(\d x) - \mu\invar(\Gamma) \Big|  \\
		& \qquad\qquad \leq C \Exp{-c(2m+2)} \Big( 1 + \int_{\rr^d} \int_{\rr^d}V(y) P^F_s(x,\d y) \mu\invar(\d x) \Big) .
\end{align*}
Using \eqref{eq:rcl-Lyap},
\begin{align*}
		\Big|\int_{\rr^d} P^F_{2m+s}(x,\Gamma) \mu\invar(\d x) - \mu\invar(\Gamma) \Big|
			&   \leq C \Exp{-c(2m+2)} \Big( 1 + K + \int_{\rr^d} V(x) \mu\invar(\d x) \Big).
\end{align*}
But the left-hand side does not depend on~$m \in \nn$ because~$\mu\invar$ is invariant for the $2$-skeleton. We therefore have $\int P^F_{s}(x, {\cdot\,}) \mu\invar(\d x) = \mu\invar$ for all $s \geq 0$, {i.e.} that~$\mu\invar$ is stationary for the orginial continuous-time process.

Now, for any $|f| \leq 1 + V$, $s \in [0,2)$ and $m \in \nn \cup \{0\}$,
\begin{align*}
	& \Big| \int_{\rr^d} f(y) [P^F_{2m+s}(x,\d y) - \mu\invar(\d y)]\Big| \\
		&\qquad\qquad
			= \Big| \int_{\rr^d}\int_{\rr^d} f(y)  P^F_{2m}(x,\d z) P^F_{s}(z,\d y) - f(y)P^F_{s}(z,\d y)\mu\invar (\d z) \Big| \\
		&\qquad\qquad
			= \Big| \int_{\rr^d} \Big( \int_{\rr^d} f(y) P^F_s(z, \d y)\Big) [P^F_{2m}(x,\d z)-\mu\invar(\d z)]\Big|.
\end{align*}
Since $|f| \leq 1 + V$, we have by \eqref{eq:rcl-Lyap} that
\begin{align*}
	\Big| \int_{\rr^d} f(y) P^F_s(z, \d y)\Big|
		&\leq \int_{\rr^d} (1 + V(y)) P^F_s(z, \d y)
		\leq (K+1) \big( 1 +  V(z) \big).
\end{align*}
Therefore, we may apply~\eqref{eq:m-bound-disc} with $f$ replaced by $\tfrac{1}{K+1}\int f(y) P^F_s({\cdot\,}, \d y)$ to get
\begin{align*}
	\Big| \int_{\rr^d} f(y) [P^F_{2m+s}(x,\d y) - \mu\invar(\d y)]\Big|
		&\leq (K+1) C \Exp{-c(2m+2)} \big( 1+ V(x) \big).
\end{align*}
Because any time~$t > 0$ can be written as $2m + s$ with $s \in [0,2)$, this is\,---\,up to a relabeling of the constants\,---\,the assertion of Theorem~\ref{thm:main-new}.

	\section{Networks of oscillators}\label{sec:net}
		We introduce the mathematical description of important physical systems that our main result covers, from the simplest to the most intricate. Based on~\cite{JPS17}, we also discuss the assumptions~\textnormal{(K)},~\textnormal{(D)} and~\textnormal{(G)} of our main result in this context. Discussion of the weak H\"ormander condition~\textnormal{(H)} is postponed to the next section.

		\subsection{The linear chain coupled to Langevin thermostats}
		Consider $L$ unit masses, each labelled by an index in~$\{1,2,\dotsc,L-1,L\}$ and whose position is restricted to a line. For $i=1,2,\dotsc, L-1$, the $i$th mass is attached to the $(i+1)$th mass by a spring of spring constant~$k > 0$. Each mass is also pinned by a spring of spring constant $\kappa \geq 0$. The position coordinate~$q_i$ of the $i$th mass is measured relative to a rest position~$q_i^\textnormal{eq}$; see Figure~1. Perturbations of this system are described by Hamiltonians of the form
		\begin{align*}
			h : \rr^L \oplus \rr^L &\to \rr \\
				(p,q) &\mapsto \frac 12 \sum_{i=1}^L p_i^2 + \frac 12 \sum_{i=1}^{L}  \kappa q_i^2 + \frac 12 \sum_{i=1}^{L-1}  k (q_{i+1} - q_i)^2 + U(q)
		\end{align*}
		where $U \in C^\infty(\rr^L; \rr)$ is a perturbing potential.

		Coupling the $1$st and $L$th oscillator to Langevin heat baths at positive temperatures~$\theta_1$ and~$\theta_L$ with positive coupling constants~$\gamma_1$ and~$\gamma_L$ yields the equations of motion
		\begin{align*}
			\d q_i &= p_i \d t,  && \hspace{-0.5in} 1 \leq i\leq L, \\
			\d p_i &= -[\kappa q_i + k (q_i - q_{i-1}) - k (q_{i+1} - q_i) + \partial_{i}U(q)] \d t, 			&& \hspace{-0.5in} 1 < i < L, \\
			\d p_1 &= -[\kappa q_1 - k (q_{2} - q_1) + \partial_{1}U(q)] \d t
				- \gamma_1 p_1 \d t + \sqrt{2\gamma_1 \theta_1} \d W_{1,t},  && \\
			\d p_L &= -[\kappa q_L + k (q_{L} - q_{L-1}) + \partial_{L}U(q)] \d t
				- \gamma_L p_L \d t + \sqrt{2\gamma_L \theta_L} \d W_{L,t}, &&
		\end{align*}
		where~$(W_{1,t})_{t\geq 0}$ and~$(W_{L,t})_{t\geq 0}$ are independent 1-dimensional Wiener processes.

		\begin{figure}
			\centering
			\usetikzlibrary{decorations.pathmorphing}
			\usetikzlibrary{patterns}
			\begin{tikzpicture}[scale=0.75, every node/.style={scale=0.75},
						declare function={
	     			noise = rnd;
	     			langevin(\t) = noise;
	     	}]

				\draw [fill= blue!14] (8,0.5) rectangle (12,-0.5);
			\draw[ultra thin,blue!60,x=2,y=10,shift={(8cm,-0.17cm)}] (0,0) -- plot [domain=0:56.9, samples=1044, smooth] (\x,{langevin(\x)});
				\node at (12.25,0) {$\theta_L$};

				\draw [fill=orange!22] (-2,0.5) rectangle (2,-0.5);
			\draw[ultra thin,orange,x=2,y=20,shift={(-2cm,-0.35cm)}] (0,0) -- plot [domain=0:56.9, samples=1044, smooth] (\x,{langevin(\x)});
				\node at (-2.25,0) {$\theta_1$};

				\draw[gray] (-2,0.15) -- (12,0.15);
				\draw[gray] (-2,-0.15) -- (12,-0.15);

				\draw[dashed,gray] (3,-1) -- (3,0.6);
				\draw[dashed,gray] (3.2,-1) -- (3.2,0.6);
				\draw[->] (3,0.5) -- (3.2,0.5);
				\node at (3.1,.75) {$q_2$};

				\node[draw,fill=white] (v1) at (0.8,0) {$\quad 1 \quad $};
				\node[draw,fill=white] (v2) at (3.2,0) {$\quad 2 \quad$};
				\node[fill=white] (v5) at (5.2,0) {$\ \dotsb \quad $};
				\node[draw,fill=white] (v3) at (7.2,0) {$\, L-1 \,$};
				\node[draw,fill=white] (v4) at (9,0) {$\quad L \quad$};

				\draw[gray] (4,-0.15) -- (6,-0.15);
				\draw[gray] (4,0.15) -- (6,0.15);

				\draw[white,pattern=north east lines] (-2,-1) rectangle (12,-1.5);
				\draw (-2,-1) -- (12,-1);

				\node at (3.35,-0.6) {$\kappa$};
				\node at (4.1,0.35) {$k$};

				\draw [decorate,decoration={coil,segment length=2.5pt}](v1) -- (1,-1);
				\draw [decorate,decoration={coil,segment length=2.5pt}](v2) -- (3,-1);
				\draw [decorate,decoration={coil,segment length=2.5pt}](v3) -- (7,-1);
				\draw [decorate,decoration={coil,segment length=2.5pt}](v4) -- (9,-1);
				\draw [decorate,decoration={coil,segment length=6pt}](v1) -- (v2);
				\draw [decorate,decoration={coil,segment length=2pt}](v3) -- (v4);
				\draw [decorate,decoration={coil,segment length=4pt}](v2) -- (v5);
				\draw [decorate,decoration={coil,segment length=4pt}](v5) -- (v3);
			\end{tikzpicture}
			\caption{Depiction of the linear harmonic chain where the $1$st and $L$th oscillator are connected to heat baths at temperatures $\theta_1$ and $\theta_L$ respectively.}
			\label{fig:springs}
		\end{figure}

		This system can be put into the form~\eqref{eq:SDE} with $d = 2L$ and $n = 2$ by setting
		$$
			X = \left(\begin{matrix} \begin{matrix} \\ p \\ \ \end{matrix} \\  \begin{matrix} \\ q \\ \ \end{matrix} \end{matrix}\right), \qquad
			A = \left(\begin{matrix} \begin{smallmatrix} -\gamma_1 & 0 & & & \\ 0 & 0 & & & \\ 0 & 0 & & & \\ & & \ddots & & \\ & & & 0 & 0 \\  & & & 0 & 0 \\ & & & 0 & -\gamma_L \end{smallmatrix}  & \begin{smallmatrix} -k-\kappa& k & & & \\ k & -2k-\kappa & & & \\  0 & k & & & \\ & & \ddots & & \\ & & & k & 0 \\ & & & -2k-\kappa & k  \\ & & & k & -k-\kappa \end{smallmatrix} \\ \one & \begin{matrix} &&&&\\&&0&& \\&&&& \end{matrix} \end{matrix}\right),
		$$
		$$
		B = \left(\begin{matrix} \begin{smallmatrix} \sqrt{2\gamma_1\theta_1} & 0 \\ 0 & 0 \\ \vdots& \vdots\\ 0 & 0  \\ 0 & \sqrt{2\gamma_L\theta_L} \end{smallmatrix}  \\  \begin{matrix} &&&&\\&&0&& \\&&&& \end{matrix} \end{matrix}\right)
		$$
		and
		$
			F(X) = F(p,q) = -\nabla_q U(q)
		$.

		The Kalman condition~\textnormal{(K)} is met for the pair~$(A,B)$ (with $d_* \leq L$) as soon as $k > 0$ and the eigenvalues of~$A$ then have strictly negative real part (condition~\textnormal{(D)} holds)~\cite{JPS17}.

		The growth condition~\textnormal{(G)} on the vector field~$F$ in the general setting is to be imposed on the gradient~$\nabla_q U$ of the perturbing potential~$U$ for the chain of oscillators: we require that it is Lipschitz and that there exists $\power \in [0,\tfrac{1}{2d_*})$ such that $|\nabla_q U(q)| = O(1 + |q|)^\power$ as $|q| \to \infty$. This potential is {not restricted} to one-body (pinning) or two-body interaction terms; it can for example include a sum of  bounded three-body interaction terms.

		\subsection{More general geometries in the Langevin regime}\label{ssec:lang}
		Let $I$ be a finite set and distinguish a nonempty subset $J \subset I$ of the sites, where the thermal noise will act. Fix a temperature $\theta_j > 0$ for the bath associated to each site $j \in J$. We can then generalize the above model to different geometries and different spring constants by considering
		\begin{equation}\label{eq:lang-mat}
			X = \left(\begin{matrix} p \\ \omega q \end{matrix}\right), \qquad
			A = \left(\begin{matrix} -\frac 12 \iota \iota^* & -\omega^* \\ \omega & 0 \end{matrix}\right), \qquad
			B = \left(\begin{matrix} \iota \\ 0 \end{matrix}\right) \vartheta^{1/2},
		\end{equation}
		and $F(p,\omega q) = -\nabla_q U(q)$,
		where
		\begin{align*}
			\omega : \rr^I \to \rr^I,
		\end{align*}
		is a nonsingular linear map\footnote{We use the symbol~$\omega$ for the linear map encoding the frequencies of the system in order to ease the comparison with other works to which we refer. Unfortunately, $\omega$ is also standard notation for elements of the underlying probability space. We trust that the meaning of the symbol is clear from the context.} and where $\vartheta$ and $\iota$ are of the form
		\begin{align*}
			\vartheta: \qquad \rr^{J} &\to \rr^{J} \\
				(u_j)_{j \in J} &\mapsto (\theta_j u_j)_{j \in J},
		\end{align*}
		and
		\begin{align*}
			\iota: \qquad \rr^{J} &\to \rr^{I} \\
				(u_j)_{j \in J} &\mapsto (\sqrt{2 \gamma_j} u_j)_{j \in J} \oplus 0_{I \setminus J}.
		\end{align*}
		Again, $\gamma_j$ is the coupling constant for the $j$th oscillator of the boundary. More explicitly, the equations of motion then take the familiar form
		\begin{align*}
			\d q &= p \d t, \\
			\d p &= -\omega^*\omega q \d t  - \nabla_{q} U(q) \d t - \tfrac 12 \iota \iota^* p \d t + \iota \vartheta^{1/2} \d W_t .
		\end{align*}

		Lemma 4.1 in~\cite{JPS17} states that if the pair $(\omega^*\omega, \iota)$ satisfies the Kalman condition~\textnormal{(K)}, then the pair $(A,B)$ defined by~\eqref{eq:lang-mat} also satisfies the Kalman condition. By Theorem~5.1(2) there, it then immediately implies the dissipativity condition~\textnormal{(D)}. In Section~4.1 there, the case of the triangular network is treated and explicit sufficient conditions for the Kalman condition are given in terms of the spring constants.
		Again, the growth condition~\textnormal{(G)} is to be imposed on the gradient~$\nabla_q U$ of the pertrubing potential~$U$.

		As mentioned in the introduction, the recent work of Cuneo, Eckmann, Hairer and Rey-Bellet \cite{CEHRB} provides a result of existence, uniqueness and exponentially fast convergence in a similar setup. Their conditions C3--C5 on the behaviour of the potential at infinity are significantly less restrictive than our conditions~\textnormal{(D)} and~\textnormal{(G)}, allowing for strong anharmonicity. However, their nondegeneracy condition~C2 is needed in all points of the phase space while our H\"ormander condition~\textnormal{(H}) is only needed in one point. Their controllability condition~C1 on the topology of the graph plays a role similar to that of our Kalman condition~\textnormal{(K)}.

		\subsection{Coupling through additional degrees of freedom}\label{ssec:add-dof}
			As pointed out {e.g.} in~\cite{JPS17}, models where the noise acts through auxiliary degrees of freedom enjoy the same structural properties, and are thus also suitable for our framework. We refer the reader to~\cite{FKM65,Tr77,EPR99a} for discussions of the physical interpretation and derivation of such models. Because of these auxiliary degrees of freedom, the model is sometimes said to be semi-Markovian.

			Let~$I$ and~$J$ be finite sets as above and consider $X = (r, p, \omega q) \in \rr^{J} \oplus \rr^I \oplus \rr^I$ for some nonsingular linear map $\omega : \rr^I \to \rr^I$.
			In addition, let $\Lambda: \rr^{J} \to \rr^I$ be a linear injection and let $\iota : \rr^{J} \to \rr^{J}$ and $\vartheta : \rr^{J} \to \rr^{J}$ be linear bijections. We set
			\begin{equation}\label{eq:add-dof-mat}
				A = \left(\begin{matrix}
							-\frac{1}{2} \iota \iota^* & -\Lambda^* & 0 \\
							\Lambda & 0 & -\omega^* \\
							0 & \omega & 0
						\end{matrix}\right)
			\qquad\text{
			and
			}\qquad
				B = \left(\begin{matrix}
						\iota \\ 0 \\ 0
					\end{matrix}\right)
					\vartheta^{1/2};
			\end{equation}
			the important structural constraints are
			\begin{gather}
				\qquad \vartheta > 0, \qquad\qquad\qquad\qquad\qquad\quad
				B^* B > 0, \\
				\ker (A - A^*) \cap \ker B^* = \{0\}, \qquad\qquad
				A + A^* = -B \vartheta^{-1} B^*.
			\end{gather}
			The perturbation~$F$ is taken to be of the form
			$$
				F: X = (r,p,\omega q) \mapsto -\nabla_q U(q)
			$$
			for some smooth potential $U : \rr^I \to \rr$ encoding the anharmonic part of both the interaction and the pinning potential. More explicitly, the equations of motion then read
			\begin{align*}
				\d q &= p \d t, \\
				\d p &= - \omega^* \omega q \d t - \nabla_{q} U(q)\d t + \Lambda r(t) \d t, \\
				\d r  &= - \tfrac{1}{2} \iota^*\iota r \d t - \Lambda^*  p \d t - \iota \vartheta^{1/2} \d W_t.
			\end{align*}

			\begin{proposition}\label{prop:K-add}
				If the pair $(\omega^*\omega, \Lambda)$ satisfies the Kalman condition, then the pair $(A,B)$ also satisfies the Kalman condition~\textnormal{(K)}.
			\end{proposition}

			\begin{proof}
				 Let $(\hat r, \hat p, \omega \hat q)$ be a target for the system in time $T > 0$. If $(\omega^*\omega,\Lambda)$ satisfies the Kalman condition, there exists $\eta_1 \in C^1_0([0,T];\rr^n)$ such that the solution $(p_1(t), q_1(t))$ of
				\begin{align*}
					\dot p_1 &= -\omega^* \omega q_1 + \Lambda \dot\eta_1, 	& p_1(0) &= 0,\\
					\dot q_1 &= p_1,			& q_1(0) &= 0,
				\end{align*}
				satisfies $(p_1(T), q_1(T)) = (\hat p, \hat q)$. Note that $(\tfrac tT \hat r,p_1(t), q_1(t))$ is then a solution of the system
				\begin{align*}
					\dot r_2 &= - \tfrac 12 \iota \iota^* (r_2 + \dot\eta_2) - \Lambda^* p_2 + \iota \vartheta^{1/2} \dot\zeta_2,  & r_2(0) &= 0,\\
					\dot p_2 &= -\omega^* \omega q_2 + \Lambda (r_2 + \dot\eta_2), & p_2(0) &= 0,\\
					\dot q_2 &= p_2 & q_2(0) &= 0,
				\end{align*}
				for the choices of control
				\begin{align*}
					\eta_2(t)  &= \eta_1(t) - \int_0^{t }\tfrac sT \hat r \d s, \\
					\zeta_2(t) &= (\iota \vartheta^{1/2})^{-1}\int_0^t \tfrac 1T \hat r + \tfrac 12 \iota \iota^* \dot\eta_1(s) + \Lambda^* p_1(s) \d s,
				\end{align*}
				hitting the prescribed target at time~$t = T$.

				Finally, note that with $\tilde\eta$ a smooth approximation of $\dot\eta_2$ that is $0$ at times $t=0$ and $t=T$, $(r_2(t) + \tilde\eta(t),p_2(t),q_2(t))$ is an approximate solution of
				\begin{align*}
					\dot r &= -\tfrac 12 \iota \iota^* r - \Lambda^* p + \iota \vartheta^{1/2} \dot\zeta, & r(0) &= 0,\\
					\dot p &= - \omega^* \omega q + \Lambda r , & p(0) &= 0,\\
					\dot q &= p, & q(0) &= 0,
				\end{align*}
				for the choice of control
				\begin{align*}
					\zeta(t) = \zeta_2(t) + (\iota \vartheta^{1/2})^{-1}\tilde\eta(t).
				\end{align*}
				Therefore, the original system is approximately controllable from 0. Because the system is linear, we conclude that the pair $(A,B)$ satisfies the Kalman condition.
			\end{proof}

			Then, Theorem~5.1(2) of~\cite{JPS17} states that, in this setup, the Kalman condition~\textnormal{(K)} implies that all the eigenvalues of~$A$ have strictly negative real part, {i.e.} condition~\textnormal{(D)}.

			In particular, for~$A$ and~$B$ arising from a pair~$(\omega^*\omega,\Lambda)$ satisfying the Kalman condition~\textnormal{(K)}, as long as $|\nabla_q U(q)| = O(1+|q|)^\power$ as $|q| \to \infty$, and as long as there exists a point where the weak H\"ormander condition holds,  the field $q \mapsto \omega^* \omega q +\nabla_q U(q)$ is allowed to be degenerate in nonnegligible regions of the position space. This is to be compared the nondegeneracy hypothesis~\textnormal{H2)} in~\cite{EPR99a,RBT02,Ca07} and \textnormal{C2} in~\cite{CEHRB} that are needed everywhere.

\section{The weak H\"ormander condition}\label{sec:Hor}

	As a starting point, we note that under the assumption~\textnormal{(K)}, the condition~$\textnormal{(H)}$ is automatically satisfied for any~$F$ with compact support or any~$F$ whose derivatives up to order $d - 1$ vanish at a point. Also note that a standard perturbative argument shows that if the conditions~\textnormal{(D)},~\textnormal{(K)} and~\textnormal{(G)} are satisfied, then there exists $\lambda_0 > 0$ such that the system
	$$
		\d X_t  = A X_t \d t + \lambda F(X_t)\d t + B \d W_t
	$$
	admits a unique invariant measure satisfying~\eqref{eq:exp-conv} as soon as $0 < \lambda < \lambda_0$.

	A more subtle perturbative argument is presented in Proposition~\ref{prop:pert-H}. We then give an example of a physically motivated potential to which this proposition applies in the context of networks of oscillators.

	In view of the definition of the weak H\"ormander condition, we are interested in the part of the tangent space spanned by Lie derivatives. The Lie derivatives $\L_G b$, $\L_G^2 b$, $\dotsc$, $\L_G^{d_*-1}b$ with $G: x \mapsto Ax +F(x)$ and $b$ a constant vector field will play a particularly important role. A direct computation shows
	\begin{align*}
		\L_G b &= -DG[b], \\
		\L_G^2 b
			&= +DG^2[b] - D^2G[b,G], \\
		\L_G^3 b &= -DG^3[b] + 2DG[D^2G[b,G]]  + D^2G[DG[b],G] - D^3G[b,G,G] - D^2G[b,DG[G]],
	\end{align*}
	and so forth. Here, the point of the space at which the vectors fields are taken is implicit and we use
   $$
      D^j G [{\cdot\,},{\cdot\,}, \dotsc,{\cdot\,}] : \underbrace{\rr^d \times \rr^d \times \dotsb \times \rr^d}_{j \text{ times}} \to \rr^d
   $$
   for the $j$th Fr\'echet derivative of the map $G : \rr^d \to \rr^d$ at this point. The above pattern generalises in the following way.

\medskip

\paragraph{Claim.} The difference between $\L_G^k b $ and $(-1)^k{DG}^k[b]$ is a linear combination over $\zz$ of compositions of Fr\'echet derivatives of~$G$ with~$b$. In each term, $b$ appears once, $G$ appears~$N_0$ times, $DG$ appears~$N_1$ times, $\dotsc$, $D^kG$ appears~$N_k$ times, with $N_1 \neq k$ and
	\begin{gather}\label{eq:k-sum-order}
		\sum_{j=0}^k  N_j =
		\sum_{j=0}^k j  N_j = k.
	\end{gather}

\begin{proof}
	We proceed by induction on~$k$. For $k = 1$ we have
	\begin{align*}
		\mathcal{L}_G b = -DG [b],
	\end{align*}
   which satisfies the claim. Assume now that the result holds for some $k \in \nn$ so that~$\mathcal{L}^k_G b - (-1)^k DG^k[b]$ is a sum of terms satisfying~\eqref{eq:k-sum-order}.
	Since
	$$
		\mathcal{L}^{k+1}_G b =  - DG [ \mathcal{L}^{k}_G b] + D(\mathcal{L}^{k}_G b)[G],
	$$
	the first term yields $-(-1)^k DG [DG^k [b]]$ and terms with the same form as those of $\mathcal{L}^{k}_G b$, but with the changes $k \mapsto k+1$ (adding $ N_{k+1} = 0$) and $ N_1 \mapsto  N_1 + 1$. It indeed satisfies the right condition on the $N$'s if $\mathcal{L}^{k}_G b$ does.
	As for the second term, by the product rule, each term in~$\mathcal{L}^{k}_G b$ yields a sum of terms undergoing $ N_0 \mapsto  N_0 + 1$ and
	$ N_j \mapsto  N_{j} - 1$ and $ N_{j+1} \mapsto  N_{j+1} + 1$ for one and only one~$j \in \{1, \dotsc, k\}$.
\end{proof}

\begin{proposition}\label{prop:pert-H}
	Suppose that the pair $(A,B)$ satisfies the Kalman condition~\textnormal{(K)} and that there exists a sequence $(y^{(n)})_{n\in\nn}$ in~$\rr^d$ that is bounded away from 0 and such that
	$$
		\lim_{n \to \infty} |y^{(n)}|^{k-1} \|D^k F(y^{(n)})\|  = 0
	$$
	for each $k = 1, 2, \dotsc, d_* - 1$. Then, there exists a point~$x_0 \in \rr^d$ where the weak H\"ormander condition~\textnormal{(H)} is satisfied.
\end{proposition}

\begin{proof}
	Let $G$ denote $y \mapsto A y + F(y)$ and let $b$ stand for a column of~$B$. By our previous claim, we have the bound
	\begin{align*}
		&|(\mathcal{L}^k_G b)(y) - (-1)^k(DG(y))^k[b]| \\
		&\qquad \qquad \leq \sum_{N \in \mathcal{A}} |C_N| |b||G(y)|^{ N_0} \|D G(y)\|^{ N_1} \|D^2G(y)\|^{ N_2} \dotsb \|D^k G(y)\|^{ N_k} \\
		&\qquad \qquad \leq
			\sum_{N \in \mathcal{A}} |C_N| |b|(\|A\||y| + \tfrac 18 \|M\|^{-1}|y| + c_1)^{ N_0}   \|D G(y)\|^{ N_1} \|D^2G(y)\|^{ N_2} \dotsb \|D^k G(y)\|^{ N_k}
	\end{align*}
	where
	$
		\mathcal{A} := \{ N = ( N_0,  N_1, \dotsc,  N_k) \in (\nn \cup \{0\})^k \text{ satisfying~\eqref{eq:k-sum-order} and } N_1 \neq k\}
	$
	and $C_N$ is a combinatorial factor in~$\zz$.

 	By condition~\eqref{eq:k-sum-order},
	\begin{align*}
		|y|^{ N_0} \|D G(y)\|^{ N_1} \|D^2G(y)\|^{ N_2} \dotsb \|D^k G(y)\|^{ N_k}
			&= |y|^{\sum_{j' = 2}^k (j'-1) N_{j'}} \prod_{j=1}^{k} \|D^{j} G(y)\|^{ N_{j}}  \\
			&= \|DG(y)\|^{ N_1} \prod_{j=2}^{k} |y|^{(j-1) N_j}\|D^jG(y)\|^{ N_j}.
	\end{align*}
	Along the subsequence~$(y^{(n)})_{n \in \nn}$ in the hypothesis, for each $j \geq 2$,
	$$
		\lim_{n \to \infty} |y^{(n)}|^{j-1} \|D^j G(y^{(n)})\| =	\lim_{n \to \infty} |y^{(n)}|^{j-1} \|D^j F(y^{(n)})\|  = 0.
	$$
	In the case $j=1$, we have
	$$
		\limsup_{n \to \infty} \|D G(y^{(n)})\| \leq \|A\| + \limsup_{n\to\infty}\|D F(y^{(n)})\| = \|A\|.
	$$
	Therefore,
	$$\lim_{n \to \infty} |(\mathcal{L}^k_G b)(y^{(n)}) - (-1)^k(DG(y^{(n)}))^k[b]| = 0$$
	for $k = 1, 2, \dotsc, d_*-1$, and for $n$ large enough,
	$$
		\operatorname{span} \{ b, \mathcal{L}_G b, \dotsc, \mathcal{L}_G^{d_*-1} b \}_{y = y^{(n)}} = \operatorname{span} \{ b, DG(y) b, \dotsc, (DG(y))^{d_* - 1}b \}_{y = y^{(n)}}.
	$$
	Finally note that
	$$
		\lim_{n \to \infty}\|(DG(y^{(n)}))^k b - A^k b\| \leq \limsup_{n \to \infty} \sum_{j=1}^{k} \binom{k}{j} \|A\|^{k-j} \|DF(y^{(n)})\|^{j} = 0.
	$$
	We conclude from the Kalman condition that for $N \in \nn$ large enough
	$$
		 \operatorname{span} \{ b, DG(y) b, \dotsc, (DG(y))^{d_* - 1}b : b \in \operatorname{ran} B \}_{y = y^{(N)}}
	$$
	coincides with
	$$
		\operatorname{ran} \{ B, A B, \dotsc, A^{d_* - 1}B \} = \rr^d.
	$$
   The results holds with $x_0 = y^{(N)}$.
\end{proof}

\begin{example}\label{ex:coul}
 	Consider that the masses in the models of Section~\ref{sec:net}, although restricted to a single spatial degree of freedom, live in 3-dimensional space and each hold an electric charge of Gaussian density
	$$
		\rho_i({\,\cdot\,}) = \frac{ Q}{(2\pi)^{3/2}\sigma^3}\exp\Big( - \frac{|{\,\cdot\,} - (q_i + q_i^\textnormal{eq})|^2}{2\sigma^2}\Big)
	$$
	where $\sigma$ is a parameter with dimension of length and $Q$ is the electric charge of each mass. In view of Poisson's equation in~$\rr^3$, this gives rise to the term
	$$
		U(q) = \sum_{i \in I} \sum_{\substack{i' \in I \\ i' \neq i}} \frac{Q^2}{4\pi\epsilon_0|(q_i + q_i^{\textnormal{eq}})-(q_{i'} + q_{i'}^{\textnormal{eq}})|} \frac{2}{\sqrt{\pi}} \int_0^{\frac{|(q_i + q_i^{\textnormal{eq}})-(q_{i'} + q_{i'}^{\textnormal{eq}})|}{\sqrt 2 \sigma}} \Exp{- s^2} \d s
	$$
	in the Hamiltonian.
	This potential satisfies the condition of the previous proposition: take for example a sequence with $q^{(n)}_i = i n \sigma$.

	For the sake of matching exactly the setup of~\cite{EPR99a,RBT02,Ca07}, consider that $I = \{1,\dotsc,L\}$ and $J = \{1,L\}$ and that only nearest neighbours interact through the Coulomb force. Let us use the shorthand $\tilde q_i := q_i + q_i^\textnormal{eq}$. Then, the corresponding perturbing potential
	$$
		U^\textnormal{n.n.}(q) = \sum_{i =1}^{L-1} \frac{Q^2}{4\pi\epsilon_0|\tilde q_i- \tilde q_{i+1}|} \frac{2}{\sqrt{\pi}} \int_0^{\frac{|\tilde q_i - \tilde q_{i+1}|}{\sqrt 2 \sigma}} \Exp{- s^2} \d s
	$$
	also satisfies the hypotheses of our previous proposition. However, note that
	{\small
	$$
		\partial_{q_2}\partial_{q_3} U^\textnormal{n.n.}(q) = \frac{Q^2}{4\pi^{\frac{3}{2}}\epsilon_0\sigma^3} \Bigg( -\frac{4 \int_0^{\frac{|\tilde q_2-\tilde q_3|}{\sqrt 2 \sigma}} \Exp{- s^2} \d s}{|\tilde q_2-\tilde q_3|^3/\sigma^3}+\frac{2 \sqrt{2} e^{-\frac{|\tilde q_2-\tilde q_3|^2}{2 \sigma ^2}}}{|\tilde q_2-\tilde q_3|^2/\sigma^2} + {\sqrt{{2}} \e^{-\frac{|\tilde q_2-\tilde q_3|^2}{2 \sigma ^2}}} \Bigg)
	$$
	}%
	does not have a definite sign. Hence, for large values of~$Q^2 \sigma^{-3}$ (very concentrated charge distribution), the uniform condition~\textnormal{H2)} in~\cite{EPR99a,EPR99b,RBT02,Ca07} is not satisfied.
\end{example}

\appendix

\section{Decomposability properties}
\label{sec:app}
We devote this appendix to the decomposability properties of~$\ell$ in Remark~\ref{rem:dec}. We consider the case $T=1$ and $n=1$ but the argument can be easily adapted to the general case. Although we use results from the theory of Gaussian measures to show the decomposability properties, these properties are not specific to Gaussian processes and can be proved for other types of noises.

The Wiener process restricted to the interval~$[0,1]$ is a nondegenerate Gaussian measure on the Banach space $C_0([0,1];\rr)$. It has as its Cameron--Martin space the space $W^{1,2}_0([0,1];\rr)$ equipped with the inner product $$\braket{\eta,\zeta}_{W^{1,2}_0} = \int_0^1 \dot\eta(s)\dot\zeta(s) \d s.$$ This Hilbert space
 has orthonormal basis~$\{\psi_m\}_{m \in \nn}$ where
\[
	\psi_m(t) = \int_0^t \phi_m(s) \d s
\]
and where $\{\phi_m\}_{m \in\nn}$ is a Fourier basis for $L^2([0,1];\rr)$. It is dense as a subspace of~$C_0([0,1];\rr)$ equipped with the supremum norm.

Let
\(
	F_N := \operatorname{span}\{\psi_m : m \leq N\}
\)
and let~$F'_N$ be the closure in $C_0([0,1];\rr)$ of the linear span of $\{\psi_m : m > N\}$.
These sequences of subspaces satisfy~(i) and provide a decomposition $F_N \oplus F'_N$: any $\eta \in C_0([0,1];\rr)$ can be written in a unique way as $\eta_N + \eta'_N$ with $\eta_N \in F_N$ and $\eta'_N \in F'_N$. To this decomposition are associated the projectors $\Pi_N$ and $\Pi'_N$.

By the general theory of Gaussian measures~(see {e.g.}~\cite[\S{3.5}]{Bog}), Brownian motion can be represented as the almost surely convergent sum
\[
	W_t(\omega) = \sum_{m \leq N} \Xi_m(\omega)\psi_m(t) + \sum_{m > N} \Xi_m(\omega)\psi_m(t),
\]
where $(\Xi_m)_{m\in\nn}$ is a sequence of independent scalar standard normal random variables. The two sums are independent and provide the decomposition~(ii) of~$\ell$ as the product of the projected laws. Property~(iii) clearly holds.

These abstract results from the theory of Gaussian measures do not provide strong convergence of~$\Pi_N$ to the identity operator on the Banach space~$C_0([0,1];\rr)$ as $N \to \infty$ (or boundedness of the set of norms $\{\|\Pi_N\| : N \in \nn\}$, which is used in~\cite{Sh17}). However, we have the following weaker convergence result for regular enough sets of functions.

\begin{lemma}\label{lem:not-quite-bounded}
	If~$Q$ is a subset of $C_0([0,1];\rr)$ that is bounded in the norm induced by the inner product~$\braket{\,\cdot\,,\cdot\,}_{W^{1,2}_0}$, then
	\[
	\lim_{N \to\infty} \sup_{\eta \in Q} \|\eta - \Pi_N \eta \|_{C_0} = 0.\]
\end{lemma}

\begin{proof}
	First note that by construction of the basis,
	\[
		\sum_{m \in \nn} \|\psi_m\|_{C_0}^2 <\infty.
	\]
	For $\eta \in W^{1,2}_0([0,1];\rr)$, the decomposition into the two subspaces can be made explicit:
	\[
		\eta(t) = \sum_{m \leq N} \psi_m(t) \int_0^1 \phi_m(s)\dot\eta(s)\d s + \sum_{m > N} \psi_m(t) \int_0^1 \phi_m(s)\dot\eta(s)\d s
	\]
	and by the Cauchy--Schwarz inequality
	\begin{align*}
		\|\eta - \Pi_N \eta\|_{C_0}
			&\leq \Big(\sum_{m > N} \|\psi_m\|_{C_0}^2\Big)^{\frac 12}\Big(\sum_{m > N} \Big|\int_0^1 \phi_m(s)\dot\eta(s)\d s\Big|^2\Big)^{\frac 12} \\
			&\leq  \Big(\sum_{m > N} \|\psi_m\|_{C_0}^2\Big)^{\frac 12} \|\eta\|_{W^{1,2}_0}.
	\end{align*}
	The convergence thus follows from the hypothesis $\sup_{\eta \in Q} \| \eta \|_{W^{1,2}_0} < \infty$.
\end{proof}

\bibliographystyle{amsalpha}
\bibliography{sde-harm}

\end{document}